\documentclass[a4paper,11pt]{article}

\usepackage{amsmath,amsthm,amssymb,latexsym,indentfirst,eucal}

\usepackage[T1]{fontenc}
\usepackage[latin9]{inputenc}
\usepackage{mathrsfs}

\usepackage{newcent}

\usepackage[margin=3cm]{geometry}

\newtheorem{theorem}{Theorem}[section]
\newtheorem{proposition}[theorem]{Proposition}
\newtheorem{lemma}[theorem]{Lemma}
\newtheorem{corollary}[theorem]{Corollary}
\newtheorem{remark}[theorem]{Remark}

\newtheorem{assumption}[theorem]{Assumption}

\newcommand{\sq}{\Box}

\begin{document}

\title{Optimal investment under behavioural criteria --
a dual approach}
\author{Mikl\'os R\'asonyi\thanks{MTA Alfr\'ed R\'enyi Institute of Mathematics, Budapest and University of Edinburgh,
e-mail: rasonyi@renyi.mta.hu}\and 
Jos\'e G. Rodr\'\i guez-Villarreal\thanks{University of Edinburgh,
e-mail: {J.G.Rodriguez-Villarreal@sms.ed.ac.uk}}}

\date{\today}

\maketitle

\begin{abstract}
We consider a discrete-time, generically incomplete market model and 
a behavioural investor with power-like utility and distortion
functions. The existence of optimal strategies in this setting has been shown in \cite{CarRas11}
under certain conditions on the parameters of these power functions. 

In the present paper we prove the existence of optimal strategies under a different set
of conditions on the parameters, identical to the ones in \cite{RR11}, which were shown to be
necessary and sufficient in the Black-Scholes model. We also relax some assumptions of \cite{CarRas11}.

Although there exists no natural dual problem for optimisation under behavioural criteria (due to the lack
of concavity), we will rely on techniques based on the usual duality between attainable
contingent claims and equivalent martingale measures. 
\end{abstract}



\section{Introduction}

This paper complements and improves results of \cite{CarRas11} where the existence of an optimal strategy for an
investor with behavioural criteria was proved under certain parameter restrictions (Assumption \ref{p}b below). 
Here we show the same
result under different restrictions on the parameters (Assumption \ref{p}a) which are identical to the ones in 
\cite{RR11} but they are neither stronger nor weaker 
than Assumption \ref{p}b. Assumption \ref{p}a is necessary and sufficient
in certain continuous-time models (this is shown in \cite{RR11} except a borderline case whose proof is yet unpublished). 
Furthermore, we manage to reprove the main result of \cite{CarRas11} under somewhat weaker assumptions.

The key new ideas are imported from \cite{RR11} and rely on the construction of an equivalent martingale measure for the 
price process whose density has nice integrability properties (see Lemma \ref{construct} below). It is this martingale
measure that permits us to prove the tightness of an optimising sequence of strategies (Lemma \ref{tajtt} below), that's why
we call our approach a dual one. 

\section{Model description}\label{222}

Fix an  integer $T>0$ acting as time horizon in the sequel and
a filtered probability space $\left(\Omega,\mathcal{F},\left\{ \mathcal{F}_{t}\right\} _{t=0,\ldots,T},\mathbb{P}\right)$.
We consider a financial market evolving in discrete time
consisting of $d$ risky assets whose discounted prices are
given by an $\mathbb{R}^{d}$-valued adapted stochastic process, $S=\left(S_{t}\right)_{t=0,\ldots,T}$
where $S_{t}=\left(S_{t}^{1},\ldots,S_{t}^{d}\right)$.

In addition, we are assuming the
financial market to be liquid and frictionless, that is, all costs
and constraints associated with transactions are non-existent, investors
are allowed to short-sell stocks and to borrow money, and it is always
possible to buy or sell an unlimited number of shares of any asset.

We denote by $\Xi_{t}^{d}$ the set of $d-$dimensional $\mathcal{F}_{t}-$measurable
random variables. Let $\mathscr{W}$ be the set of $\mathbb{R}$-valued (or $\mathbb{R}^d$-valued) random variables
$Y$ such that $\mathbf{E}_{\mathbb{P}}\left|Y\right|^{p}<\infty$
for all $p>0$. 

Trading stategies
are characterised by an initial capital $z$ and a $d$-dimensional process
$\left\{ \theta_{t}:1\leqslant t\leqslant T\,\right\}$ representing the holdings in the respective assets. We assume $\theta$ to be 
predictable, i.e. $\theta_t\in\Xi_{t-1}^{d}$ for all $t$. The class of all such
strategies is denoted by $\Phi$. 

We define $X_{t}^{z}\left(\theta\right):=z+\sum_{k=1}^{t}\theta_{k}\cdot\Delta S_{k}$,
the value process of a portfolio with initial investment $z$ and trading
strategy $\theta$, where $\Delta S_k:=S_k-S_{k-1}$ and $\cdot$ denotes scalar product. For $x\in\mathbb{R}$ the notations
$x_+,x_-$ stand for positive and negative parts, respectively.

\begin{assumption}\label{sna} For all $t\geqslant1$, $\Delta S_{t}\in\mathscr{W}$.
Furthermore, for $0\leqslant t\leqslant T-1,$ there exist $\mathcal{F}_t$-measurable $\kappa_{t},\,\beta_{t}>0$
satisfying $\frac{1}{\kappa_{t}},\,\frac{1}{\beta_{t}}\in\mathscr{W}$
such that 
\begin{equation}\label{ssna}
\mathrm{ess}\inf_{\xi\in\Xi_{t}^{d}}\mathbb{P}\left(\xi\cdot\Delta S_{t+1}\leqslant-\kappa_{t}\left|\xi\right|\left|\mathcal{F}_{t}\right.\right)\geqslant\beta_{t}\,\,\mbox{a.s.}
\end{equation}
\end{assumption}
We may and will assume $\kappa_t,\beta_t\leqslant 1$ in the sequel. 
As pointed out in \cite{CarRas11}, \eqref{ssna} is a strengthened form
of the absence of arbitrage condition.
We denote by $\mathcal{M}^{e}\left(S\right)$ the set of equivalent martingale measures for $S$. 
Recall that, under the standard no arbitrage hypothesis, $\mathcal{M}^{e}\left(S\right)\neq \emptyset$, see e.g. \cite{JS98}.
Assumption \ref{sna} will allow us to construct a particular $\mathbb{Q}\in \mathcal{M}^{e}\left(S\right)$
with favourable properties, see Lemma \ref{construct} below.

Now we turn to the description of an economic agent.
Her attitude towards gains and loses will be described in terms
of functions $u_{+}$ and $u_{-}$. In addition, she will be assumed to distort
the ``real world'' distributions (probabilities) by means of functions
$w_{+}$ and $w_{-}.$ She will further have a ``benchmark''
or reference point $B$ which is used when evaluating portfolio payoffs
at the terminal time $T$.

\begin{assumption}\label{u} We assume that $u_{\pm}:\mathbb{R}^{+}\rightarrow\mathbb{R}^+$
and $w_{\pm}:\left[0,1\right]\rightarrow\left[0,1\right]$ are measurable
functions such that $u_{\pm}\left(0\right)=0,\, w_{\pm}\left(0\right)=0,\, w_{\pm}\left(1\right)=1,$
and 
\begin{equation}
u_{+}\left(x\right)\leqslant k_{+}\left(x^{\alpha}+1\right),
\end{equation}
\begin{equation}
k_{-}\left(x^{\beta}-1\right)\leqslant u_{-}\left(x\right),
\end{equation}
\begin{equation}
w_{+}\left(p\right)\leqslant g_{+}p^{\gamma},
\end{equation}
\begin{equation}\label{neggi}
w_{-}\left(p\right)\geqslant g_{-}p^{\delta},
\end{equation}
with $\alpha,\beta,\gamma,\delta>0$, $k_{\pm},g_{\pm}>0$
fixed constants.
\end{assumption}

\begin{assumption}\label{p} This concerns the parameters involved in Assumption \ref{u}. For convenience, we shall consider two separate cases.
\\ 
\noindent{Assumption} \ref{p}a. The parameters $\alpha,\beta,\gamma$ and $\delta$ satisfy 
\begin{equation}\label{pa}
\alpha<\beta\mbox{ and }\frac{\alpha}{\gamma}<1<\frac{\beta}{\delta}.
\end{equation}

\noindent{Assumption} \ref{p}b. The parameters $\alpha,\beta,\gamma$ and $\delta$ are such that 
\begin{equation}\label{pb}
\delta\leq 1,\ \alpha<\beta,\mbox{ and }\frac{\alpha}{\gamma}<\beta.
\end{equation}
\end{assumption}

\begin{assumption}\label{b} Similarly to the previous assumption, we consider two cases. 
\\ 
\noindent Assumption \ref{b}a. The reference point $B\in\Xi_{T}^{1}$ belongs to $L^{1+r}\left(\mathbb{P}\right)$
for some $r>0$.\\
 
\noindent Assumption \ref{b}b. For the reference point $B\in\Xi_{T}^{1}$ there is a trading strategy $\phi\in\Phi$ and 
initial capital $b\in\mathbb{R}$ satisfying
\begin{equation}
X_{T}^{b}\left(\phi\right)=b+\sum_{t=1}^{T}\phi_{t}\cdot\Delta S_{t}\leqslant B.
\end{equation}
\end{assumption}

Clearly, in Assumption \ref{p}b above, one of the two conditions $\alpha<\beta$ and $\alpha/\gamma<\beta$
subsumes the other, depending on whether $\gamma>1$ or $\gamma\leq 1$.
An economic interpretation can easily be given to Assumption \ref{b}b. It means that the losses occurring in the 
behavioural investor's benchmark are comparable to the value of some self-financing portfolio. 

Given a real-valued random variable $X$ representing the outcome of an investment, a behavioural agent measures her satisfaction distorting
the expected utility of profits as well as the expected ''dissatisfaction'' of losses. 
Consider the nonlinear
functionals $V_{+}\left(X\right)$ and $V_{-}\left(X\right)$ defined below. Let

\begin{equation}
V_{+}\left(X\right):=\int_{0}^{\infty}w_{+}\left(\mathbb{P}\left(u_{+}\left(X_{+}\right)>y\right)\right)dy.\label{eq:behavioural functional 1}
\end{equation}
Notice that $V_{+}$ incorporates the utility of the investor on gains and $w_{+}$ produces a non-linear alteration of the 
given probability distribution.
If $w_{+}\left(x\right)=x$ then we return to the expected utility framework since in this case $V_{+}\left(X\right)=\mathbf{E}u_{+}\left(\left(X-B\right)_{+}\right)$.
Similarly, let
\begin{equation}
V_{-}\left(X\right):=\int_{0}^{\infty}w_{-}\left(\mathbb{P}\left(u_{-}\left(X_{-}\right)>y\right)\right)dy,\label{eq:behavioural function 2}
\end{equation}
and, finally, the objective or performance functional we aim to optimise is defined by 
\begin{equation}
V\left(X\right):=V_{+}\left(X\right)-V_{-}\left(X\right),\label{eq:behavioural functional}
\end{equation}
provided that at least one of the summands is finite.

According to the cummulative prospective theory (CPT) developed in
\cite{KT79} and \cite{tk}, behavioural
investors assess their satisfaction from a given portfolio at terminal
time $T$ by means of the functional defined in (\ref{eq:behavioural functional})
and the benchmark $B$. So we define the functionals $V_{+},\, V_{-}$ below by 
\begin{equation}
V_{+}\left(z,\theta_{1},\ldots,\theta_{T}\right):=V_{+}\left(X_{T}^{z}\left(\theta\right)\right)=\int_{0}^{\infty}w_{+}\left(\mathbb{P}\left(u_{+}\left(\left(X_{T}^{z}\left(\theta\right)-B\right)_{+}\right)>y\right)\right)dy,
\end{equation}
\begin{equation}
V_{-}\left(z,\theta_{1},\ldots,\theta_{T}\right):=V_{-}\left(X_{T}^{z}\left(\theta\right)\right)=\int_{0}^{\infty}w_{-}\left(\mathbb{P}\left(u_{-}\left(\left(X_{T}^{z}\left(\theta\right)-B\right)_{-}\right)>y\right)\right)dy.
\end{equation}
We say that a trading strategy $\theta\in\Phi$
is admissible for initial capital $z$ if $V_{-}\left(X_{T}^{z}\left(\theta\right)\right)<\infty$.
We denote the set of such trading strategies by $\mathcal{A}\left(z\right)$ and define, for $\theta\in\mathcal{A}(z)$,
\[
V\left(z,\theta_{1},\ldots,\theta_{T}\right):=V\left(X^{z}_T\left(\theta\right)\right)=V_{+}\left(z,\theta_{1},\ldots,\theta_{T}\right)-V_{-}\left(z,\theta_{1},\ldots,\theta_{T}\right).
\]

The optimal portfolio problem for a behavioural investor consists in finding 
$\theta^{\star}=\left(\theta_{1}^{\star},\ldots,\theta_{T}^{\star}\right)\in\mathcal{A}(z)$
such that
\begin{equation}
\sup_{\theta\in\mathcal{A}\left(z\right)}V\left(z,\theta_{1},\ldots,\theta_{T}\right)=V\left(z,\theta_{1}^{\star},\ldots,\theta_{T}^{\star}\right).
\label{eq:optimal behaviour}
\end{equation}

\section{Main results}

As it is well known, most discrete-time market models are incomplete, i.e. $\mathcal{M}^e(S)$ is not a singleton, hence 
the problem of how to choose a suitable equivalent martingale measure $\mathbb{Q}$ arises. 

\begin{lemma}\label{construct}
Under Assumption \ref{sna} there exists $\mathbb{Q}\in\mathcal{M}^e(S)$ such that for $\rho:=d\mathbb{Q}/d\mathbb{P}$ we have
both $\rho,1/\rho\in\mathscr{W}$.
\end{lemma}
\begin{proof}
We rely on \cite{RS05}, which provides a utility maximisation framework where
the existence of a martingale measure with desirable properties can be guaranteed. Define the continuously
differentiable, concave function
\begin{equation}\label{mj}
{U}\left(x\right)=\begin{cases}
x-\frac{1}{2} & \mbox{\,\,\ensuremath{\mbox{if }}}x\geqslant0\\
-\frac{1}{2}\left(x-1\right)^{2} & \,\,\mbox{if}\, x<0.
\end{cases}
\end{equation}
The hypotheses of Proposition 7.1 in \cite{RS05} hold by Assumption \ref{sna} and by \eqref{mj}, hence there is  
$\mathbb{Q}\in\mathcal{M}^e(S)$ such that 
$$
\rho=\frac{d\mathbb{Q}}{d\mathbb{P}}=\frac{U'(X^{0}_T(\phi^*))}{EU'(X^{0}_T(\phi^*))}
$$
for some $\phi^*\in\Phi$. Inspecting the proof of Proposition 7.1 in \cite{RS05} one can easily check that 
$\phi_t^*\in\mathscr{W}$ for all $t$. Hence $\rho\in\mathscr{W}$ and
$\rho$ is bounded away from $0$, a fortiori, $1/\rho\in\mathscr{W}$. 
\end{proof}
We fix the probability $\mathbb{Q}$ just constructed for later use. It will be key in establishing moment estimates which underlie our main results.
Note also that, under Assumption \ref{b}a, $B\in L^{1+\epsilon}\left(\mathbb{Q}\right)$ for all $0<\epsilon<r$
by H\"older's inequality and $\rho\in\mathscr{W}$.

We first address the well-posedness of the optimal portfolio problem for a behavioural
investor. We say that the optimal investment problem (\ref{eq:optimal behaviour}) is \emph{well-posed} if the supremum in (\ref{eq:optimal behaviour}) 
is finite. If the supremum is infinite then the problem is called \emph{ill-posed}.

We know from section 3 of \cite{CarRas11} that $\alpha/\gamma\leqslant\beta/\delta$ and $\alpha<\beta$ are necessary for well-posedness.
It is an open problem whether they are sufficient as well. We show below, however, that either \eqref{pa} or \eqref{pb} 
are sufficient.

\begin{theorem}\label{wpa}
Under Assumptions \ref{sna}, \ref{u}, \ref{p}a and \ref{b}a, the optimisation
problem (\ref{eq:optimal behaviour}) is well-posed. In other words, 
\begin{equation}
\sup_{\theta\in\mathcal{A}\left(z\right)}V\left(z,\theta_{1},\ldots,\theta_{T}\right)<\infty.
\end{equation}
\end{theorem}

We shall use the auxiliary results given below which were shown in \cite{RR11} (see Lemmas 3.12, 3.13 and 3.14 there). 
We include their statements for the sake of completeness. 

\begin{lemma}\label{suti}
If $a,b$ and $s$ are positive numbers satisfying $\frac{b}{sa}>1$
then there exists a constant $D$ such that
\begin{equation}
\mathbf{E}_{\mathbb{P}}\left(X^{s}\right)\leqslant 1+D\left(\int_{0}^{\infty}\mathbb{P}\left(X^{b}>y\right)^{a}dy\right)^{\frac{1}{a}}.
\end{equation}
for all non-negative random variables $X$. $\sq$
\end{lemma}

\begin{lemma}\label{moz1}  
Let $d\mathbb{Q}/d\mathbb{P},d\mathbb{P}/d\mathbb{Q}\in\mathscr{W}$, $\alpha<\beta$ and $\frac{\alpha}{\gamma}<1<\frac{\beta}{\delta}$. Fix $m\in\mathbb{R}$.
Then there is some $\eta>0$ satisfying $\eta<\beta$, $\alpha<\eta$
and $\delta<\eta$, and there exist constants $L_1=L_{1}(m)$
and $L_2=L_{2}(m)$ such that 
\begin{equation}
\int_{0}^{\infty}\mathbb{P}\left(\left(X_{+}\right)^{\alpha}>y\right)^{\gamma}dy\leqslant L_{1}+L_{2}\int_{0}^{\infty}\mathbb{P}\left(\left(X_{-}\right)^{\eta}>y\right)^{\delta}dy,
\end{equation}
for all random variables $X$ with $\mathbf{E}_{\mathbb{Q}}\left[X\right]=m$. $\sq$ 
\end{lemma} 

\begin{lemma}\label{moz2} 
Let $a,\, b$ and $s$ be strictly positive real numbers such that
$s<a<b$ and $s\leqslant1$. Then there exist $0<\zeta<1$ and constants $R_{1},$ $R_{2}$ such that 
\begin{equation}
\int_{0}^{\infty}\mathbb{P}\left(X^{a}>y\right)^{s}dy\leqslant R_{1}+R_{2}\left[\int_{0}^{\infty}\mathbb{P}\left(X^{b}>y\right)^{s}dy\right]^{\zeta},
\end{equation}
for all non-negative random variables $X$. $\sq$
\end{lemma}

\begin{remark} Note that in the paper \cite{RR11} it was assumed that $u_{\pm},w_{\pm}$ are
power functions (and not only comparable to power functions as in Assumption \ref{u} above). Furthermore, 
$\alpha,\beta,\gamma,\delta\leq 1$ were stipulated, in line with the literature.
One can check in \cite{RR11} that the proof of Lemma \ref{moz1} above goes
through without this restriction.
\end{remark}

These Lemmas allow us to prove Theorem \ref{wpa}.
\begin{proof}[Proof of Theorem \ref{wpa}.] We imitate the proof of Theorem 3.15 in \cite{RR11}.
By contradiction, let us suppose that the optimisation problem is ill-posed. Then for a sequence $\phi(n)\in \mathcal{A}(z)$, $n\in\mathbb{N}$
we have
$V_{+}\!\left([X_{T}^z({\phi(n)})-B]_{+}\right)\rightarrow +\infty$ as $n \rightarrow +\infty$. 
Note that, for any non-negative $X$,
$$
V_{+}\!\left(X\right)\leqslant \int_0^{\infty} g_+ \mathbb{P}(X^{\alpha}>(y/k_+)-1)^{\gamma}dy\leq \int_0^{\infty} g_+k_+ \mathbb{P}(X^{\alpha}>t)^{\gamma}dt+g_+k_+.
$$
Thus it follows from Lemma \ref{moz1} (with the choice $m:=z-\mathbf{E}_{\mathbb{Q}}[B]$) that
\begin{equation*}
\lim_{n\rightarrow +\infty} \int_{0}^{+\infty} \mathbb{P}\!\left([X_T^z({\phi(n)})-B]_{-}^{\eta}>y\right)^{\delta} dy =+\infty
\end{equation*}
for some $\eta$ satisfying $\eta<\beta$, $\alpha<\eta$ and $\delta<\eta$. 
Notice that
\begin{equation}\label{horhos}
V_-(X)\geqslant \int_0^{\infty} g_- \mathbb{P}(k_-X^{\beta}-k_->y)^{\delta}dy\geq \int_1^{\infty} g_- k_- \mathbb{P}(X^\beta>t)^{\delta}dt.
\end{equation}
Consequently, we can apply Lemma \ref{moz2} to conclude that also $$
\lim_{n\rightarrow +\infty} V_{-}\!\left([X_T^z({\phi(n)})-B]_{-}\right)=+\infty.
$$
Therefore, using Lemmas \ref{moz1} and \ref{moz2} again (and recalling that $0<\zeta<1$),
\begin{eqnarray*}
& & V\!\left(X^z_T({\phi(n)})-B\right)\leq g_+k_+(L_{1}+1)+g_+ k_+ L_{2} 
\int_{0}^{+\infty} \mathbb{P}\!\left(([X_T^z({\phi(n)})-B]_{-})^{\eta}>y\right)^{\delta} dy \\
&-& V_{-}\!\left([X_T^z({\phi(n)})-B]_{-}\right)\leq g_+k_+\left(L_{1}+1+ L_2 R_{1}\right)\\ &+& g_+ k_+ L_{2} R_{2}
\left[\frac{V_{-}\!\left([X_T^z({\phi(n)})-B]_{-}\right)}{g_-k_-}+1\right]^{\zeta} - V_{-}\!\left([X_T^z({\phi(n)})-B]_{-}\right) \xrightarrow[n\rightarrow +\infty]{} -\infty,
\end{eqnarray*}
which is absurd. Hence, as claimed, the problem is well-posed.
\end{proof}

We present a result about well-posedness under the alternative conditions Assumptions \ref{p}b and \ref{b}b as well. 
It is worth pointing out that while the conclusions of Theorems \ref{wpa} and \ref{wpb} are identical, the
methods for proving them are significantly different. 

\begin{theorem}\label{wpb}
Under Assumptions \ref{sna}, \ref{u}, \ref{p}b and \ref{b}b the problem is well-posed, i.e.  
\begin{equation}
\sup_{\theta\in\mathcal{A}\left(z\right)}V\left(z,\theta_{1},\ldots,\theta_{T}\right)<\infty.
\end{equation}
\end{theorem}
\begin{proof} Notice that $\delta\leq 1$ and \eqref{neggi} imply the fourth inequality in Assumption 4.1 of 
\cite{CarRas11}. Hence our
result follows from Theorem 4.4 in \cite{CarRas11}. Note that in \cite{CarRas11} $\alpha,\beta,\gamma\leq 1$ were also 
assumed. As already indicated in Remark 4.2 of \cite{CarRas11}, the proofs go through 
without this restriction.\end{proof}

From now on, the existence of optimal strategies will be our main concern. We will need to assume that the filtration
is rich enough in the sense of Assumption \ref{fried} below. This Assumption means that investors 
randomize their strategies or, from a mathematical point of view, that we enlarge the underlying probability space.
We will comment on this in section \ref{adde} as well.

\begin{assumption}\label{fried}
Define $\mathcal{G}_0=\{\emptyset,\Omega\}$, and $\mathcal{G}_t=\sigma(Z_1,\ldots,Z_t)$ for $1\leq t\leq T$,
where the $Z_i$, $i=1,\ldots, T$ are $\mathbb{R}^N$-valued
independent random variables. $S_0$ is constant, $\Delta S_t$ is $\mathcal{G}_t$-adapted and $B$ is
$\mathcal{G}_T$-measurable.

Furthermore, $\mathcal{F}_t=\mathcal{G}_t\vee\mathcal{F}_0$, $t\geq 0$,
where $\mathcal{F}_0=\sigma({\varepsilon})$ with ${\varepsilon}$
uniformly distributed on $[0,1]$ and independent of $(Z_1,\ldots,Z_T)$.
\end{assumption}

\begin{remark} The above Assumption clearly implies that 
$\Delta S_{t}=f^{\left(t\right)}\left(Z_{1},...,Z_{t}\right)$ for some Borel functions $f^{(t)}$, for all
$t$, and $B=g_B\left(Z_{1},\ldots,Z_{T}\right)$ for some Borel function $g_B$.
We may and will suppose without loss of generality that each of the
$Z_{i}$ is bounded.
\end{remark}

In \cite{RR11} the existence of optimal strategies was shown under Assumption \ref{p}a 
(and $B\in L^1(\mathbb{Q})$ for some reference probability $\mathbb{Q}\in\mathcal{M}^e(S)$) 
in 
a (narrow) class of continuous-time
models. In \cite{CarRas11} existence was shown under Assumptions \ref{p}b, \ref{b}b and \ref{fried} in discrete-time models
assuming also the continuity of $f^{(t)},g_B$. In the present paper we shall
prove existence of an optimiser in discrete-time models under Assumption \ref{fried} and either Assumption \ref{p}a or Assumption \ref{p}b,
and we do not need continuity of $f^{(t)},g_B$. We first present some preparatory results.

\begin{proposition}\label{tajt} Let Assumptions \ref{sna}, \ref{u}, \ref{p}a and \ref{b}a hold and take $\mathbb{Q}\sim\mathbb{P}$ as constructed in Lemma 
\ref{construct}. Further, suppose that a sequence of trading strategies $\left\{ \theta^{n}\right\} \subset\mathcal{A}\left(z\right)$ 
satisfies
\begin{equation}
\sup_{n}V_{-}\left(z,\theta_{1}^{n},\ldots,\theta_{T}^{n}\right)<\infty.\label{eq:V- hypothesis}
\end{equation}
Then there exists $\pi>1$ such that 
\begin{equation}
\sup_{n}\mathbf{E}_{\mathbb{Q}}\left(X_{T}^{z}\left(\theta^{n}\right)\right)_{-}^{\pi}<\infty,\label{eq:xT minus}
\end{equation}
and 
\begin{equation}
\sup_{n}\mathbf{E}_{\mathbb{Q}}\left(X_{T}^{z}\left(\theta^{n}\right)\right)_{+}<\infty.\label{eq:xT plus}
\end{equation}

It follows also that  
\begin{equation}
\sup_{n}\mathbf{E}_{\mathbb{Q}}\left[\sup_{t\leqslant T}\left(X_{t}^{z}\left(\theta^{n}\right)\right)_{-}^{\pi}\right]<\infty,\label{eq:xT minus sup}
\end{equation}
\begin{equation}
\sup_{n,t}\mathbf{E}_{\mathbb{Q}}\left[\left|X_{t}^{z}\left(\theta^{n}\right)\right|\right]<\infty.\label{eq:xt abs bound}
\end{equation}
\end{proposition}

\begin{proof} This is a direct application of Lemma \ref{moz2}. Indeed, choose
$1<s<\frac{\beta}{\delta}$ and $\lambda$ such that $1<\lambda<s<\frac{\beta}{\delta}$.
Applying H\"older's inequality,
\begin{eqnarray*}
\mathbf{E}_{\mathbb{Q}}\left[\left(X_{T}^{z}\left(\theta^{n}\right)-B\right)_{-}^{\frac{s}{\lambda}}\right]=&
\mathbf{E}_{\mathbb{Q}}\left[\rho^{1/\lambda}\frac{1}{\rho^{1/\lambda}}\left(X_{T}^{z}\left(\theta^{n}\right)-B\right)_{-}^{\frac{s}{\lambda}}\right]&\leqslant\\
&\leqslant C\mathbf{E}_{\mathbb{P}}\left[\left(X_{T}^{z}\left(\theta^{n}\right)-B\right)_{-}^{s}\right]^{1/\lambda},&
\end{eqnarray*}
where $C=\mathbf{E}_{\mathbb{Q}}\left[\rho^{q/\lambda}\right]^{1/q}<\infty$ and $q$ is the conjugate number of $\lambda$. 
Lemma \ref{suti} yields that, for all $n$, 
\begin{equation}
C\mathbf{E}_{\mathbb{P}}\left[\left(X_{T}^{z}\left(\theta^{n}\right)-B\right)_{-}^{s}\right]^{1/\lambda}\leqslant C
\left(1+D\left(\int_{0}^{\infty}\mathbb{P}\left(\left(X_{T}^z\left(\theta^{n}\right)-B\right)_{-}^{\beta}>y\right)^{\delta}dy\right)^{1/\delta}\right)^{1/\lambda}
\end{equation} 
for some $D<\infty$. Hence \eqref{eq:V- hypothesis} and \eqref{horhos}
imply (\ref{eq:xT minus}), setting $\pi:=\min\{\frac{s}{\lambda},1+\frac{r}{2}\}$ (note that, as we have pointed
out after Lemma \ref{construct}, $\mathbf{E}_{\mathbb{Q}}|B|^{1+(r/2)}<\infty$). 
Moreover, H\"older's inequality gives $$
\sup_{n}\mathbf{E}_{\mathbb{Q}}\left(X_{T}^{z}\left(\theta^{n}\right)\right)_{-}<\infty.$$ 
It follows from 
Theorem 2 in \cite{JS98}  that $\left\{ X^{z}_{t}\left(\theta^{n}\right)\right\} _{t\leqslant T}$ is a martingale under $\mathbb{Q}$, 
thus $$
\mathbf{E}_{\mathbb{Q}}\left|X^z_{t}\left(\theta^{n}\right)\right|\leqslant\mathbf{E}_{\mathbb{Q}}\left|X_{T}^z\left(\theta^{n}\right)\right|,
$$
for all $n,t$. From $\mathbf{E}_{\mathbb{Q}}\left[X_{t}^z\left(\theta^{n}\right)\right]=z$
and (\ref{eq:xT minus}) we have 
\[
\sup_{n}\mathbf{E}_{\mathbb{Q}}\left(X_{T}^z\left(\theta^{n}\right)\right)_{+}\leqslant |z|+\sup_{n}\mathbf{E}_{\mathbb{Q}}\left(X_{T}^z\left(\theta^{n}\right)\right)_{-}<\infty
\]
Hence $\sup_{n}\mathbf{E}_{\mathbb{Q}}\left|X_{T}\left(\theta^{n}\right)\right|<\infty$
and this implies \eqref{eq:xT plus} as well as (\ref{eq:xt abs bound}).
In order to prove (\ref{eq:xT minus sup}), Doob's inequality is applied, noting that $f\left(x\right)=x_{-}$ is convex and hence the process 
$\{\left(X_{t}^z\left(\theta^{n}\right)\right)_{-}\}_{t\leqslant T}$ is a positive submartingale. 
\end{proof}

Notice that we could show Proposition \ref{tajt} only in a discrete time and finite horizon setting since it relies 
on Theorem 2 of \cite{JS98} which
fails in more general (e.g. continuous-time)  settings. 

\begin{remark}
In \cite{RR11} admissible strategies $\theta$ were required to satisfy \emph{both} $V_-(X_T^z(\theta))<\infty$ and 
the martingale property 
for $X^z_t(\theta)$ (under some fixed $\mathbb{Q}\in\mathcal{M}^e(S)$). 
The proof above shows that, in the present discrete-time setting, $V_-(X_T^z(\theta))<\infty$ \emph{implies}
the martingale property for $X^z_t(\theta)$ under $\mathbb{Q}$. So the domain of optimisation in the present
paper is the same as the one in \cite{RR11}.
\end{remark}

\begin{remark}\label{royal} Let $\theta=\left(\theta_{1},\theta_{2},\ldots,\theta_{T}\right)\in\mathcal{A}\left(z\right)$ be as in Proposition \ref{tajt}.
Clearly, 
\begin{equation}
\left(\theta_{T}\cdot\Delta S_{T}\right)_{+}\leqslant\left(X_{T}^z\left(\theta\right)\right)_{+}+\left(X_{T-1}^z\left(\theta\right)\right)_{-}.
\end{equation}

Thus 
\[
\mathbf{E}_{\mathbb{Q}}\left(\theta_{T}\cdot\Delta S_{T}\right)_{+}\leqslant\mathbf{E}_{\mathbb{Q}}\left(X_{T}^z\left(\theta\right)\right)_{+}+
\mathbf{E}_{\mathbb{Q}}\left(X_{T-1}^z\left(\theta\right)\right)_{-}
\]
implies
\begin{equation}
\sup_{n}\mathbf{E}_{\mathbb{Q}}\left[\left(\theta_{T}^{n}\cdot\Delta S_{T}\right)_{+}\right]<\infty.\label{eq:remark3-1}
\end{equation}
\end{remark}

We will now proceed to proving that trading
strategies satisfying \eqref{eq:V- hypothesis} have some uniformly bounded moments.

\begin{lemma}\label{tajtt} Let $\left\{ \theta^{n}\right\} \subset\mathcal{A}\left(z\right)$
be a sequence of trading strategies. Let Assumptions \ref{sna}, \ref{u}, \ref{p}a and \ref{b}a hold and assume that 
\eqref{eq:V- hypothesis} holds. Then 
\begin{equation}
\sup_{n}\mathbf{E}_{\mathbb{Q}}\left|\theta_{t}^{n}\right|^{1/2}<\infty\,\,\mbox{for }\, t=1,2,\ldots,\, T.
\end{equation}
\end{lemma}

\begin{proof} Brackets $\langle\cdot,\cdot\rangle$ are used to denote scalar product in this proof. 
A uniform bound for $\mathbf{E}_{\mathbb{Q}}\left[\left(\theta_{T}^{n}\cdot\Delta S_{T}\right)_{+}\right]$
can obtained as in Remark \ref{royal}. Using the same idea for $t\leqslant T$, 

\begin{eqnarray}\nonumber
\mathbf{E}_{\mathbb{Q}}\left\langle\theta^n_{t},\Delta S_{t}\right\rangle_{+}&\leqslant\mathbf{E}_{\mathbb{Q}}\left(X_{t}^z\left(\theta^n\right)\right)_{+}+
\mathbf{E}_{\mathbb{Q}}\left(X_{t-1}^z\left(\theta^n\right)\right)_{-}\leqslant\\
\label{qwertz} &\leqslant \mathbf{E}_{\mathbb{Q}}\left|X_{t}^z\left(\theta^n\right)\right|+
\mathbf{E}_{\mathbb{Q}}\left|X_{t-1}^z\left(\theta^n\right)\right|,\label{eq:theta S_t}
\end{eqnarray}
and the right-hand side is bounded uniformly in $n$ by Proposition \ref{tajt}.

Denote $\rho_t:=\mathbf{E}_{\mathbb{P}}[\rho |\mathcal{F}_t]$ for $0\leq t\leq T$. From Assumption \ref{sna}, 
\begin{eqnarray}
\mathbf{E}_{\mathbb{Q}}\left\langle\theta_{T}^n,\Delta S_{T}\right\rangle_{+} &\geqslant& \mathbf{E}_{\mathbb{Q}}\left[\left|\theta_{T}^n\right|\left\langle
\frac{\theta^n_{T}}{\left|\theta^n_{T}\right|},\Delta S_{T}\right\rangle\mathbf{1}_{\left\{\left\langle\frac{\theta_{T}^n}
{\left|\theta_{T}^n\right|},\Delta S_{T}\right\rangle\geqslant\kappa_{T-1}\right\}}\right]\geqslant\\
&\geqslant&\mathbf{E}_{\mathbb{Q}}\left[\left|\theta^n_{T}\right|\kappa_{T-1}\mathbb{Q}\left(\left.\left\langle
\frac{\theta^n_{T}}{\left|\theta_{T}^n\right|},\Delta S_{T}\right\rangle\geqslant\kappa_{T-1}\right|\mathcal{F}_{T-1}\right)\right].
\end{eqnarray}
Using the property 
\[
\mathbf{E}_{\mathbb{Q}}\left[\left.\eta\right|\mathcal{G}\right]=\mathbf{E}_{\mathbb{P}}\left[\left.\eta\frac{d\mathbb{Q}}{d\mathbb{P}}\right|\mathcal{G}\right]/\mathbf{E}_{\mathbb{P}}\left[\left.\frac{d\mathbb{Q}}{d\mathbb{P}}\right|\mathcal{G}\right],
\]
of conditional expectations which holds for any sigma-algebra $\mathcal{G}$ and for any positive random
variable $\eta$, we get
\begin{equation}\label{gex}
\mathbf{E}_{\mathbb{Q}}\left\langle\theta^n_{T},\Delta S_{T}\right\rangle_{+}\geqslant
\mathbf{E}_{\mathbb{Q}}\left[\frac{\left|\theta^n_{T}\right|}{\rho_{T-1}}\kappa_{T-1}
\mathbf{E}_{\mathbb{P}}\left[\left.\mathbf{1}_{\left\{\left\langle\frac{\theta_{T}^n}{\left|\theta_{T}^n\right|},
\Delta S_{T}\right\rangle\geqslant\kappa_{T-1}\right\}}\rho_{T}\right|\mathcal{F}_{T-1}\right]\right].
\end{equation}
Denote 
$A_T=\left\{\left\langle\frac{\theta_{T}^n}{\left|\theta_{T}^n\right|},\Delta S_{T}\right\rangle\geqslant\kappa_{T-1}\right\}$ and 
apply the (conditional) Cauchy inequality to the right-hand side:
\[
\mathbf{E}_{\mathbb{P}}\left[\left.\mathbf{1}_{A_T}\rho_{T}\right|\mathcal{F}_{T-1}\right] \geqslant \mathbb{P}^{2}\left(A_T\left|\mathcal{F}_{T-1}\right.\right)\left/\mathbf{E}_{\mathbb{P}}\left[\left.\frac{1}{\rho_{T}}\right|\mathcal{F}_{T-1}\right]\right.
.
\]
From Assumption \ref{sna} and Cauchy's inequality, the right-hand side of \eqref{gex} can be minorised by 
\begin{multline}\label{E:m2}
\mathbf{E}_{\mathbb{Q}}\left[\frac{\left|\theta_{T}^n\right|}{\rho_{T-1}}\kappa_{T-1}\beta_{T-1}^{2}\left/\mathbf{E}_{\mathbb{P}}\left[\left.\rho_{T}^{-1}\right|\mathcal{F}_{T-1}\right]\right.\right]\geqslant\\ \hspace{4.5 cm}
\shoveright{\mathbf{E}_{\mathbb{Q}}\left[\left|\theta_{T}^n\right|^{1/2}\right]^2\left/\mathbf{E}_{\mathbb{Q}}\left[\rho_{T-1}\kappa_{T-1}^{-1}\beta_{T-1}^{-2}\mathbf{E}_{\mathbb{P}}\left[\left.\rho_{T}^{-1}\right|\mathcal{F}_{T-1}\right]\right]\right.,}
\end{multline}
thus 
\[
\mathbf{E}_{\mathbb{Q}}\left\langle\theta_{T},\Delta S_{T}\right\rangle_{+}\mathbf{E}_{\mathbb{Q}}\left[\rho_{T-1}\kappa_{T-1}^{-1}\beta_{T-1}^{-2}\mathbf{E}_{\mathbb{P}}\left[\left.\rho_{T}^{-1}\right|\mathcal{F}_{T-1}\right]\right]\geqslant\mathbf{E}_{\mathbb{Q}}\left[\left|\theta_{T}\right|^{1/2}\right]^2.
\]
The same procedure applies to $\theta_{t}$, $t=1,\ldots, T$ by (\ref{eq:theta S_t}).
Thus, for all $t$,
\[
\sup_{n}\mathbf{E}_{\mathbb{Q}}\left[\left|\theta_{t}^{n}\right|^{1/2}\right]\leqslant
\sup_{n}\left[\mathbf{E}_{\mathbb{Q}}\left\langle\theta_{t}^{n},\Delta S_{t}\right\rangle_{+}
\mathbf{E}_{\mathbb{Q}}\left[\rho_{t-1}\kappa_{t-1}^{-1}\beta_{t-1}^{-2}\mathbf{E}_{\mathbb{P}}\left[\left.\rho_{t}^{-1}\right|\mathcal{F}_{t-1}\right]\right]\right]^{1/2}
<\infty,
\]
by Assumption \ref{sna} and \eqref{qwertz}.\end{proof}
 
\begin{remark}
Applying H\"older's inequality, the estimates above can be carried out with no significant alteration for any $0<\xi<1$, i.e. 
\begin{equation}
\sup_{n}\mathbf{E}_{\mathbb{Q}}\left|\theta_{t}^{n}\right|^{\xi}<\infty\,\,\mbox{for }\, t=1,2,\ldots,\, T
\end{equation}
can be shown. For simplicity we did this only for $\xi=\frac{1}{2}$.
 \end{remark}

From the last Lemma the next one follows trivially.
\begin{lemma}\label{rela} Under the Assumptions \ref{sna}, \ref{u}, \ref{p}a and \ref{b}a, 
let $\left\{ \theta^{n}\right\}_{n\geqslant1} \subset\mathcal{A}\left(z\right)$ a sequence of admissible trading strategies such that $\sup_{n}V_{-}\left(z,\theta_{1}^{n},\ldots\theta_{T}^{n}\right)<\infty$.
Then $\left\{ \theta{}^{n}\right\}_{n\geqslant1} $ is a tight sequence of $\mathbb{R}^{dT}-$valued
random variables on the probability space $(\Omega,\mathcal{F},\mathbb{P})$.\hfill $\sq$
\end{lemma}

\begin{proposition}\label{kabi} Let $\left\{ \theta^{n}\right\}_{n\geqslant1} $ be a sequence of trading strategies whose set of laws is tight. 
Let $\mu_n$
be the law of $X_{T}^{z}\left(\theta^{n}\right)-B$ for all $n$. Under Assumption \ref{fried}, there exists a law $\mu^{\star}$
and a trading strategy $\theta^{\star}$ such that $\mu^{\star}=Law\left(X_{T}^{z}\left(\theta^{\star}\right)-B\right)$
and $\mu^{\star}$ is an accumulation point of the sequence $\left\{ \mu_{n}\right\}_{n\geq 1}$
in the weak (narrow) topology. 
\end{proposition}
\begin{proof} Lemma 9.4 of \cite{CarRas11} provides independent random variables $\varepsilon',\tilde{\varepsilon}$,
uniform on $[0,1]$,
which are both functions of $\varepsilon$. Now following the proof of Theorem 6.8 of \cite{CarRas11}
verbatim (with $\phi_1=\ldots=\phi_T=0$) we obtain an $\mathcal{F}_t$-predictable process $\theta^{\star}_t$
such that, by Prokhorov's theorem, the law of 
$$
Y_k:=(\varepsilon',\theta^{n_k}_1,\ldots,\theta^{n_k}_T,Z_1,\ldots,Z_T)
$$
tends to that of 
$$
Y:=(\varepsilon',\theta^{\star}_1,\ldots,\theta^{\star}_T,Z_1,\ldots,Z_T)
$$
for a subsequence $n_k$, as $k\to\infty$. Skorokhod's theorem provides random variables
$$
\bar{Y}_k=(\bar{\varepsilon}'(n_k),\bar{\theta}^{n_k}_1,\ldots,\bar{\theta}^{n_k}_T,\bar{Z}_1^{n_k},\ldots,\bar{Z}_T^{n_k})
$$
and
$$
\bar{Y}=(\bar{\varepsilon}',\bar{\theta}_1,\ldots,\bar{\theta}_T,\bar{Z}_1,\ldots,\bar{Z}_T)
$$
on some probability space such that $Law(Y_k)=Law(\bar{Y}_k)$ for all $k$, $Law(Y)=Law(\bar{Y})$ and
$\bar{Y}_k$ tends to $\bar{Y}$ a.s.

 By assumption, we also have that $\Delta S_{i}=f^{\left(i\right)}\left(Z_{1},\ldots,Z_{i}\right)$ and $B=g_B(Z_1,\ldots,Z_T)$. 
Denoting $\Delta\bar{S}_{i}=f^{\left(i\right)}\left(\bar{Z}_1,\ldots,\bar{Z}_i\right)$ and $\bar{B}:=g_B(\bar{Z}_1,
\ldots,\bar{Z}_T)$ we have that
 \[
 Law\left(\left(Z_{i}\right)_{i\leqslant T},\left(\theta_{i}^{\star}\right)_{i\leqslant T},\left(\Delta S_{i}\right)_{i\leqslant T},B\right)=
 Law\left(\left(\bar{Z_{i}}\right)_{i\leqslant T},\left(\bar{\theta_{i}}\right)_{i\leqslant T},\left(\Delta\bar{S_{i}},\right)_{i\leqslant T},\bar{B}\right).
\]
Therefore 
\begin{equation}
Law\left(\sum_{i=1}^{T}\theta_{i}^{\star}\cdot\Delta S_{i}-B\right)=
Law\left(\sum_{i=1}^{T}\bar{\theta_{i}}\cdot\Delta\bar{S}_{i}-\bar{B}\right).\label{eq:eq in law}
\end{equation}

Denote $\bar{B}^k:=g_B(\bar{Z}^{n_k}_1,\ldots,\bar{Z}^{n_k}_T)$ and $\Delta\bar{S}^k_i:=f^{(i)}(\bar{Z}_1^{n_k},\ldots,
\bar{Z}_i^{n_k})$. 
By Th\'eor\`eme 1 in \cite{BEKSY}, $\Delta\bar{S}_{i}^{k}\rightarrow\Delta\bar{S}_{i}$
for all $i$ in probability and also $\bar{B}^k\to \bar{B}$ in probability, $k\to\infty$. 

It follows that
\begin{equation}
\sum_{i=1}^{T}\bar{\theta}_{i}^{k}\cdot\Delta\bar{S}_{i}^{k}-\bar{B}^k\to\sum_{i=1}^{T}\bar{\theta_{i}}\cdot\Delta
\bar{S}_{i}-\bar{B},
\end{equation} 
in probability, hence in law. 
In other words, we have 
\begin{equation}
\sum_{i=1}^{T}\theta_{i}^{k}\cdot\Delta S_{i}-B\to\sum_{i=1}^{T}\theta_{i}^{\star}\cdot\Delta S_{i}-B
\end{equation}
in law. This finishes the proof.
\end{proof}

Our first main result on the existence of optimal strategies now follows easily from Proposition \ref{kabi} above.

\begin{theorem}\label{mta} Let Assumptions \ref{fried}, \ref{sna}, \ref{u}, \ref{p}a and \ref{b}a be in force
and let $u_{\pm}$, $w_{\pm}$ be continuous. 
Then the supremum in (\ref{eq:optimal behaviour}) is attained by an optimal strategy $\theta^{\star}$.
\end{theorem}

\begin{proof} 
Let us take a maximising sequence of admissible strategies $\left\{ \theta^{j}\right\}_{j\geqslant1}$, i.e.
$$
V\left(z,\theta_{1}^{j},\ldots,\theta_{T}^{j}\right)\rightarrow \sup_{\theta\in\mathcal{A}\left(z\right)}V\left(z,\theta_{1},\ldots,\theta_{T}\right),
\ j\to\infty.
$$
The proof of Theorem \ref{wpa} shows that we necessarily have 
$\sup_{j}V_{-}\left(z,\theta_{1}^{j},\ldots,\theta_{T}^{j}\right)<\infty$, showing (\ref{eq:V- hypothesis}).

Due to Lemma \ref{rela}, we can conclude 
that the sequence $\left\{ \theta^{j}\right\}_{j\geqslant1} $ is tight. Proposition \ref{kabi} then shows that there is a strategy 
$\theta^{\star}\in\Phi$,
$\theta^{\star}=\left(\theta_{1}^{\star},\theta_{2}^{\star},\ldots,\theta_{T}^{\star}\right)\in\mathbb{R}^{dT}$
such that $X_{T}^{z}\left(\theta^{j}\right)-B\to X_{T}^{z}\left(\theta^{\star}\right)-B$ in law (along a subsequence
which we assume to be the original sequence). Now our aim is to prove 
\begin{equation}\label{eq:limsup_claim}
\limsup_{j}V\left(z,\theta_{1}^{j},\theta_{2}^{j},\ldots,\theta_{T}^{j}\right)\leqslant V\left(z,\theta_{1}^{\star},\theta_{2}^{\star},\ldots,\theta_{T}^{\star}\right).
\end{equation}

By the continuous mapping theorem we have 
\[\left(X_{T}^{z}\left(\theta^{j}\right)-B\right)_{\pm}\rightarrow\left(X_{T}^{z}\left(\theta^{\star}\right)-B\right)_{\pm}
\]
and
\begin{equation}\label{bhjk}
u_{\pm}\left(\left(X_{T}^{z}\left(\theta^{j}\right)-B\right)_{\pm}\right)\rightarrow u_{\pm}\left(\left(X_{T}^{z}\left(\theta^{\star}\right)-B\right)_{\pm}\right)
\end{equation}
in law. Let $D$ be the set of discontinuity
points of the limiting distributions in \eqref{bhjk}. Then 
\[
\mathbb{P}\left(u_{\pm}\left(\left(X_{T}^{z}\left(\theta^{j}\right)-B\right)_{\pm}\right)\geqslant y\right)\rightarrow\mathbb{P}
\left(u_{\pm}\left(\left(X_{T}^{z}\left(\theta^{\star}\right)-B\right)_{\pm}\right)\geqslant y\right)\,\,\,\mbox{for all }y\in\mathbb{R}^{+}\backslash D,
\]
in particular, for Lebesgue-a.e. $y$. By Assumption \ref{u}, 
\begin{equation}
w_{+}\left(\mathbb{P}\left(u_{+}\left(\left(X_{T}^{z}\left(\theta^{j}\right)-B\right)_{+}\right)\geqslant y\right)\right)\leqslant 
g_+\left[\mathbb{P}\left(u_{+}\left(\left(X_{T}^{z}\left(\theta^{j}\right)-B\right)_{+}\right)\geqslant y\right)\right]^{\gamma}.
\end{equation}
Take $1/\gamma <\lambda<1/\alpha$. Applying Markov's inequality and Assumption \ref{u} again,
\begin{equation}
g_+\left[\mathbb{P}\left(u_{+}^{\lambda}\left(\left(X_{T}^{z}\left(\theta^{j}\right)-B\right)_{+}\right)\geqslant y^{\lambda}\right)\right]^{\gamma}
\leqslant \frac{c'}{y^{\lambda\gamma}}\left\{\mathbf{E}_{\mathbb{P}}\left(1+\left(X_{T}^{z}\left(\theta^{j}\right)-B\right)_{+}^{\alpha}\right)^{\lambda}\right\}^{\gamma} ,\label{eq:Chebyshev}
\end{equation}
for some $c'>0$, hence
\[
c'\left[\mathbb{P}\left(u_{+}^{\lambda\gamma}\left(\left(X_{T}^{z}\left(\theta^{j}\right)-B\right)_{+}\right)\geqslant y^{\lambda}\right)\right]^{\gamma}
\leqslant\frac{c''}{y^{\lambda\gamma}}\mathbf{E}_{\mathbb{P}}^{\gamma}\left(1+\left(X_{T}^{z}\left(\theta^{j}\right)-B\right)_{+}^{\alpha\lambda}\right).
\]
with some $c''>0$. Furthermore,
\begin{equation}
\frac{c''}{y^{\lambda\gamma}}\mathbf{E}_{\mathbb{P}}^{\gamma}\left(1+\left(X_{T}^{z}\left(\theta^{j}\right)-B\right)_{+}^{\alpha\lambda}\right)\leqslant
c'''\left(1+\left[\mathbf{E}_{\mathbb{P}}\left(X_{T}^{z}\left(\theta^{j}\right)\right)_{+}^{\alpha\lambda}\right]^{\gamma}+
\left[\mathbf{E}_{\mathbb{P}}B_{-}^{\alpha\lambda}\right]^{\gamma}\right).\label{eq:attainability bound I}
\end{equation}
The last term is finite by Assumption \ref{b}a.
H\"older's inequality applied with $p=\frac{1}{\alpha\lambda}$ gives
$$
\mathbf{E}_{\mathbb{P}}\left(X_{T}^{z}\left(\theta^{j}\right)\right)_{+}^{\alpha\lambda}=\mathbf{E}_{\mathbb{Q}}\left[\frac{1}{\rho}\cdot\left(X_{T}^{z}\left(\theta^{j}\right)\right)_{+}^{\alpha\lambda}\right]\leqslant C_{1}\left[\mathbf{E}_{\mathbb{Q}}\left(X_{T}^{z}\left(\theta^{j}\right)\right)_{+}\right]^{\alpha\lambda}
$$
with $C_{1}=\mathbf{E}_{\mathbb{Q}}\left[\rho^{1/(\alpha\lambda-1)}\right]^{1-\alpha\lambda}$.
Thus
\begin{equation*}
  w_{+}\left(\mathbb{P}\left(u_{+}\left(\left(X_{T}^{z}\left(\theta^{j}\right)-B\right)_{+}\right)
\geqslant y\right)\right)\leqslant 
 \frac{1}{y^{\lambda\gamma}}\left\{ 
 D_{1}+D_{2}\left[\mathbf{E}_{\mathbb{Q}}\left(X_{T}^{z}\left(\theta^{j}\right)\right)_{+}\right]^{\alpha\lambda\gamma}\right\},
\end{equation*}
with suitable constants $D_1,D_2$.
The condition $\sup_{j}V_{-}\left(z,\theta_{1}^{j},\ldots,\theta_{T}^{j}\right)<\infty$ 
implies that $\sup_{j}\mathbf{E}_{\mathbb{Q}}\left(X_{T}^{z}\left(\theta^{j}\right)\right)_{+}<\infty$ (see Proposition \ref{tajt}).
This in turn gives that the sequence of positive functions $w_{+}\left(\mathbb{P}\left(u_{+}\left(\left(X_{T}^{z}\left(\theta^{j}\right)-B\right)_{+}\right)\geqslant y\right)\right)$
can be dominated by $\frac{K}{y^{\lambda\gamma}}$ for some $K>0$. 

This estimate allows to apply Lebesgue's theorem since 
\begin{equation}
w_{+}\left(\mathbb{P}\left(u_{+}\left(\left(X^{z}\left(\theta^{j}\right)-B\right)_{+}\right)\geqslant y\right)\right)
\leqslant\mathbf{I}_{\left[0,1\right]}\left(y\right)+\frac{K}{y^{\lambda\gamma}}\cdot\mathbf{I}_{\left(1,\infty\right)}\left(y\right),\label{eq:fatous cond}
\end{equation}

which yields
\begin{eqnarray*}
\limsup_{j}V_{+}\left(z,\theta^{j}_1,\ldots,\theta_T^j\right)\leqslant\int_{0}^{\infty}\limsup_{j}w_{+}
\left(\mathbb{P}\left(u_{+}\left(\left(X^{z}\left(\theta^{j}\right)-B\right)_{+}\right)\geqslant y\right)\right)dy,
\end{eqnarray*}
and the latter equals $V_{+}\left(z,\theta^{\star}_1,\ldots,\theta^{\star}_T\right)$.

On the other hand, by Fatou's lemma applied to $V_{-}\left(z,\theta_{1}^{j},\ldots,\theta_{T}^{j}\right)$
we have
\[
V_{-}\left(z,\theta_{1}^{\star},\ldots,\theta_{T}^{\star}\right)\leqslant\liminf_j V_{-}\left(z,\theta_{1}^{j},\ldots,\theta_{T}^{j}\right)
\leqslant\sup_{j}V_{-}\left(x_{0},\theta_{1}^{j},\ldots,\theta_{T}^{j}\right)<\infty,
\]
so $\theta^{\star}\in\mathcal{A}\left(z\right)$ and 
\begin{gather*}
\limsup_{j}V^{+}\left(z,{\theta^{j}_1},\ldots,\theta^j_T\right)-\liminf_j V^{-}\left(z,{\theta^{j}_1},\ldots,
\theta^j_T\right)\leqslant\\
V^{+}\left(z,\theta_{1}^{\star},\ldots,\theta_{T}^{\star}\right)-V^{-}\left(z,\theta_{1}^{\star},\ldots,\theta_{T}^{\star}\right),\\
\end{gather*}
thus
\begin{gather*}
\limsup_{j}\left\{ V^{+}\left(z,{\theta^{j}_1},\ldots,\theta^j_T\right)-V^{-}\left(z,{\theta^{j}_1},\ldots,
\theta^j_T\right)\right\} \leqslant\\
V^{+}\left(z,\theta_{1}^{\star},\ldots,\theta_{T}^{\star}\right)-V^{-}\left(z,\theta_{1}^{\star},\ldots,\theta_{T}^{\star}\right),
\end{gather*}
which yields (\ref{eq:limsup_claim}). Hence $\theta^{\star}$ is optimal
and the supremum is attainable.
\end{proof}

Now we turn to the case of Assumptions \ref{p}b and \ref{b}b. 
%

\begin{theorem}\label{mtb}  Let Assumptions \ref{fried}, \ref{sna}, \ref{u}, \ref{p}b and \ref{b}b be in force and let
$u_{\pm}$, $w_{\pm}$ be continuous. Then the supremum is attained in \eqref{eq:optimal behaviour} by some $\theta^{\star}\in\mathcal{A}(z)$.
\end{theorem}

\begin{proof} 
We can follow the proof of Theorem 6.8 in \cite{CarRas11} verbatim up to the point of
constructing $$
Y:=(\varepsilon',\theta^{\star}_1,\ldots,\theta^{\star}_T,Z_1,\ldots,Z_T).
$$
Then the argument of Proposition \ref{kabi} shows that 
\begin{equation}
\sum_{i=1}^{T}\theta_{i}^{k}\cdot\Delta S_{i}-B\to\sum_{i=1}^{T}\theta_{i}^{\star}\cdot\Delta S_{i}-B
\end{equation}
in law. From this point on the Fatou-lemma argument of Theorem 6.8 in \cite{CarRas11} applies verbatim and
optimality of $\theta^{\star}$ can be established.
\end{proof}


\section{A sufficient condition}\label{adde}

Assumption \ref{fried} may look restrictive at first sight. Hence we provide a simple sufficient condition for its validity.

\begin{proposition}\label{hv}
Let $\mathcal{G}_t:=\sigma(\tilde{Z}_1,\ldots,\tilde{Z}_t)$ for $t=1,\ldots,T$ and $\mathcal{G}_0=\{\emptyset,\Omega\}$ where the
$\tilde{Z}_i$, $i=1,\ldots,T$ are $N$-dimensional random variables with a Lebesgue-a.e. positive joint density on $\mathbb{R}^{TN}$. Then
there are independent $\mathbb{R}^N$-valued random variables $Z_i$, $i=1,\ldots,T$ such that 
$\mathcal{G}_t=\sigma({Z}_1,\ldots,{Z}_t)$ for $t=1,\ldots,T$.
\end{proposition}

We first recall a foklore-type result (see Lemma 9.6 of \cite{CarRas11}).

\begin{lemma}\label{trf}
Let $X$ be a real-valued random variable with atomless law. Let
$F(x):=P(X\leq x)$ denote its cumulative distribution function. Then
$F(X)$ has uniform law on $[0,1]$.\hfill $\sq$
\end{lemma}


The following results are parallel to Lemma 9.8 and Corollary 9.9 of \cite{CarRas11}.
In the sequel, when we write ``measurable bijection'' we mean that both the function and its
inverse are measurable.

\begin{lemma}\label{end}
Let $(Y,W)$ be an $\mathbb{R}\times\mathbb{R}^k$-valued
random variable with Lebesgue almost everywhere positive density
$f(x^1,\ldots,x^{k+1})$.  Then there is a measurable bijection 
$H$ from $\mathbb{R}^{k+1}$ into $[0,1] \times \mathbb{R}^{k}$ such that
$H^i(x^1,\ldots,x^{k+1})=x^i$ for $i=2,\ldots,k+1$ and
$Z:=H^1(Y,W)$ is uniform on $[0,1]$, independent of $W$.
\end{lemma}
\begin{proof}
The conditional distribution function of $Y$ knowing
$W=(x^2,\ldots,x^{k+1})$,
\[
F(x^1,\ldots,x^{k+1}):=
\frac{\int_{-\infty}^{x^1}f(z,x^2,\ldots,x^{k+1})dz}{\int_{-\infty}^{\infty}
f(z,x^2,\ldots,x^{k+1})dz},
\]
is clearly measurable (in all its variables). By a.e. positivity of $f$, $F$ is also
strictly increasing in $x^1$ hence the function
\[
H:\, (x^1,\ldots,x^{k+1})\to (F(x^1,\ldots,x^{k+1}),x^2,\ldots,x^{k+1})
\]
is a measurable bijection. 
By Lemma \ref{trf} the conditional law $P(H^1(Y,W)\in\,\cdot\,\vert
W=(x^2,\ldots,x^{k+1}))$ is uniform on $[0,1]$ for Lebesgue-almost
all $(x^2,\ldots,x^{k+1})$, which shows that $H^1(Y,W)$ is
independent of $W$ with uniform law on $[0,1]$.
\end{proof}

\begin{corollary}\label{bank}
Let $(\tilde{W}_1,\ldots,\tilde{W}_k)$ be an $\mathbb{R}^k$-valued
random variable with a.e. positive density
(w.r.t. the $k$-dimensional Lebesgue measure). Then there are independent random
variables $W_1,\ldots,W_k$ and measurable bijections
$g_l(k):\mathbb{R}^l\to\mathbb{R}^l$, $1\leq l\leq k$ such that
$(\tilde{W}_1,\ldots,\tilde{W}_l)=g_l(k)(W_1,\ldots,W_l)$.
\end{corollary}
\begin{proof} The case $k=1$ is vacuous.
Assume that the statement is true for $k\geq 1$, let us prove it for
$k+1$. We may set $g_l(k+1):=g_l(k)$, $1\leq l\leq k$, it remains to
construct $g_{k+1}(k+1)$ and $W_{k+1}$.

We wish to apply Lemma \ref{end} in this induction step.
It provides a measurable bijection
$s:\mathbb{R}^{k+1}\to\mathbb{R}^{k+1}$ such that
$s^m(x^1,\ldots,x^{k+1})=x^m$, $1\leq m\leq k$ and
$W_{k+1}:=s^{k+1}(\tilde{W}_1,\ldots,\tilde{W}_{k+1})$ is
independent of $(\tilde{W}_1,\ldots,\tilde{W}_k)$ and hence of
$$(W_1,\ldots,W_k)=g_k(k)^{-1}(\tilde{W}_1,\ldots,\tilde{W}_k).$$
Define $a:\mathbb{R}^{k+1}\to\mathbb{R}^{k+1}$ by
\begin{eqnarray*}
a(x^1,\ldots,x^{k+1})& := & (g_k(k)^{-1}(x^1,\ldots,x^{k}),s^{k+1}(x^1,\ldots,x^{k+1}))\\
 & = &
s(g_k(k)^{-1}(x^1,\ldots,x^{k}),x^{k+1}),
\end{eqnarray*}
$a$ is clearly a measurable bijection. Notice that $a(\tilde{W}_1,
\ldots,\tilde{W}_{k+1})=(W_1,\ldots,W_{k+1})$. Set
$g_{k+1}(k+1):=a^{-1}$. This finishes the proof of the induction
step and hence concludes the proof.
\end{proof}

\begin{proof}[Proof of Proposition \ref{hv}.] Apply Corollary \ref{bank} with the choice
$k:=TN$ and $$\tilde{W}_{(t-1)N+l}:=\tilde{Z}_t^l,\quad l=1,\ldots,N$$
and $t=1,\ldots, T$. By the construction in Corollary \ref{bank}, {taking $Z_t^l:=W_{(t-1)N+l}$},  one has
$(\tilde{Z}_1,\ldots,\tilde{Z}_t)=g_{tN}(TN)(Z_1,\ldots,Z_t)$ hence
indeed $\mathcal{G}_t=\sigma(Z_1,\ldots,Z_t)$ for $t=1,\ldots,T$. 
\end{proof}

\section{Conclusions}
In this work we have proved 
two types of results concerning the optimal portfolio problem for a behavioural investor in a discrete-time setting
with power-like utility and distortion functions. 
Namely, well-posedness and the existence of optimal strategies were shown  
under two different conditions on the parameters. We managed to
remove some restrictive assumptions of \cite{CarRas11}. 
The evolution of the stock price is allowed to be very general, 
our results cover a myriad of models in discrete-time incomplete markets, see section 8 of \cite{CarRas11}
for examples. 

\noindent\textbf{Acknowledgment.} The second author gratefully acknowledges support from grants of CONACYT,
Mexico and of the University of Edinburgh. The first author thanks Laurence Carassus for pointing out
certain errors in the manuscript.

%
%
%
%
%

\end{document}